\documentclass[final]{IEEEtran}
\usepackage{amsthm,amssymb,graphicx,multirow,amsmath,color,amsfonts}
\usepackage[update,prepend]{epstopdf}
\usepackage[noadjust]{cite}
\usepackage[latin1]{inputenc}
\usepackage{tikz} 
\usepackage{bbm} 
\usepackage{pdfpages}
\usepackage{multirow}
\usepackage{comment}

\usepackage{algorithm}
\usepackage[noend]{algpseudocode}

\usepackage{balance}
\newtheorem{theorem}{Theorem}[section]

\newtheorem{lemma}{Lemma}[section]

\usepackage{subfig}
\usepackage{optidef}
\allowdisplaybreaks 
\usepackage[colorinlistoftodos,bordercolor=orange,backgroundcolor=orange!20,linecolor=orange,textsize=scriptsize]{todonotes}

\usepackage{setspace}	

\graphicspath{{figures/}}

\newcommand{\algorithmicserver}{\textbf{Global Server}}
\algdef{SE}[SERVER]{Server}{EndServer}[1]{\algorithmicserver\ #1\ \algorithmicdo}{\algorithmicend\ \algorithmicserver}%
\algtext*{EndServer}
\newcommand{\algorithmicclient}{\textbf{Client k}}
\algdef{SE}[CLIENT]{Client}{EndClient}[1]{\algorithmicclient\ #1\ \algorithmicdo}{\algorithmicend\ \algorithmicclient}%
\algtext*{EndClient}
\makeatletter
\newcommand{\brokenline}[2][t]{\parbox[#1]{\dimexpr\linewidth-\ALG@thistlm}{\strut\raggedright #2\strut}}
\makeatother

\usepackage{romanbar}

\makeatletter
\renewcommand{\fnum@figure}{Figure \thefigure}
\makeatother

\makeatletter
\renewcommand{\fnum@table}{Table \thetable}
\makeatother
\begin{document}
\include{notation}
\title{
Energy Efficient Aerial RIS: Phase Shift Optimization and Trajectory Design 
}

\author{
    \IEEEauthorblockN{ 
    Hajar El Hammouti\IEEEauthorrefmark{1},
    Adnane Saoud\IEEEauthorrefmark{1}\IEEEauthorrefmark{2},
    Asma Ennahkami\IEEEauthorrefmark{2}, El Houcine Bergou\IEEEauthorrefmark{1}
        }
    \IEEEauthorblockA{ \IEEEauthorrefmark{1}  College of Computing, Mohammed VI Polytechnic University,\\ Ben Guerir, Morocco.\\ \IEEEauthorrefmark{2}
    L2S-CentraleSup\'elec, Univ Paris-Saclay, Gif sur Yvette, France,}
   
 \{hajar.elhammouti, adnane.saoud, elhoucine.bergou\}@um6p.ma, asma.ennahkami@student-cs.fr}

\maketitle

\begin{abstract}

Reconfigurable Intelligent Surface (RIS) technology has gained significant attention due to its ability to enhance the performance of wireless communication systems. The main advantage of RIS is that it can be strategically placed in the environment to control wireless signals, enabling improvements in coverage, capacity, and energy efficiency. In this paper, we investigate a scenario in which a drone, equipped with a RIS, travels from an initial point to a target destination. In this scenario, the aerial RIS (ARIS) is deployed to establish a direct link between the base station and obstructed users. Our objective is to maximize the energy efficiency of the ARIS while taking into account its dynamic model including its velocity and acceleration along with the phase shift of the RIS. To this end, we formulate the energy efficiency problem under the constraints of the dynamic model of the drone. The studied problem is challenging to solve. To address this, we proceed as follows. First, we introduce an efficient solution that involves decoupling the phase shift optimization and the trajectory design. Specifically, the closed-form expression of the phase-shift is obtained using a convex approximation, which is subsequently integrated into the trajectory
design problem. We then employ tools inspired by economic model predictive control (EMPC) to solve the resulting trajectory optimization.
 Our simulation results show a significant improvement in energy efficiency against the scenario where the dynamic model of the UAV is ignored.
\end{abstract}
\begin{IEEEkeywords}
Aerial reconfigurable intelligent surface, economic model prediction control, energy efficiency, phase shift, sum-rate, UAV.
\end{IEEEkeywords}
\section{Introduction} \label{sec:intro}

Over the past few years, there has been a growing interest in integrating aerial unmanned vehicles (UAVs) to expand and enhance the coverage of terrestrial communication networks \cite{mozaffari2019tutorial,ndiaye2022age}. In practice, the connection between ground base stations and users can often be obstructed by buildings or obstacles, resulting in a significant reduction in coverage \cite{hammouti2018air}. To address this issue, UAVs equipped with reconfigurable intelligent surfaces (RIS) can be deployed to effectively redirect signals towards obstructed users \cite{wu2021intelligent,liaskos2018new}. This has led to the emergence of the concept of aerial RIS (ARIS) which aims to dynamically position the RIS in the 3D space to adapt to the complex wireless communication constraints \cite{ye2022nonterrestrial}.

 Several works have
studied ARIS communication systems where various performance metrics have been optimized. For example, in~\cite{lu2021aerial}, an ARIS is employed with the objective of improving the worst signal-to-noise ratio (SNR) through joint optimization of transmit power, RIS placement, and 3D passive beamforming. Similarly, in~\cite{long2020reflections}, a UAV equipped with a RIS is deployed to address both the security and energy efficiency of a wireless network where users are widely dispersed. In~\cite{jeon2022energy}, the scenario of a high aerial platform (HAP) enabled by a RIS is studied. The HAP is used to redirect backhaul signals to UAVs. To achieve this, the authors propose a joint optimization approach, while considering the positioning of the HAP-mounted RIS, as well as the phase adjustments of RIS elements to improve the energy efficiency of the system.  The case of multiple ARISs is studied in~\cite{aung2022energy} where the authors propose deep reinforcement learning to jointly optimize the ARISs placement, the
phase shift, and power control. 

The previously cited works have focused on static ARIS, where the aerial platform remains in a fixed 3D position to provide service to a set of users. 
However, another significant scenario involves employing ARIS while a UAV is in motion, engaged in a particular mission. In this dynamic context, ARIS leverages the UAV's ability to move in the 3D space to redirect the signals to a broader range of users, thereby enhancing network connectivity and coverage. In this regard, only a handful of works can be found where the trajectory of the UAV equipped with RIS is optimized along with the RIS phase shift~\cite{liu2021elevation,duo2023joint}. In~\cite{liu2021elevation}, the authors investigate a scenario in which an ARIS is used to facilitate the communication between two users. They propose an iterative approach to maximize the minimum average
achievable rate by jointly optimizing the 3D UAV trajectory and the RIS phase shift. Another scenario involving two UAVs is considered in~\cite{duo2023joint}, with one of them equipped with a RIS. The UAVs are employed to offload computing tasks from ground devices to an access point.
To achieve this objective, the authors jointly optimize the trajectories of the two UAVs, the phase shift of the ARIS, and the computation resources with the aim of improving the energy efficiency of the system. Although promising, the proposed trajectory optimization approaches do not involve the dynamic model of the UAV. Specifically, the drone is assumed to travel with a fixed speed during the entire flight which does not align with the realistic behavior of UAVs.

In this work, we design the trajectory of a UAV that travels from a starting point to a target destination while carrying a RIS. Our objective is to maximize the energy efficiency of the drone while taking into account its dynamic model including its velocity and acceleration along with the phase shift of the RIS. Our contributions can be summarized as follows. 

\begin{itemize}
    \item We formulate the energy efficiency problem by jointly optimizing the UAV's trajectory, its acceleration and velocity profiles, as well as the phase shifts during the UAV's mission.

\item  We relax the continuous-time problem into a discrete-time one. Then, to address the non-convexity of the discrete-time problem, we introduce an efficient solution that involves decoupling the phase shift optimization and the trajectory design.
\item To address the phase shift optimization, a convex approximation is applied which yields a closed-form expression of the phase shifts with respect to the UAV's position at each time step. 
\item The energy efficiency problem is then reduced to a trajectory optimization under dynamic model constraints, which is solved using tools inspired from economic model prediction control (EMPC). 
\item Our simulation results show a significant improvement in energy efficiency against the scenario where the dynamic model of the UAV is ignored. 
\end{itemize}

The remainder of this paper is organized as follows. The system model is described in section II. In section III, we formulate the studied problem as an energy efficiency maximization under the UAV's dynamic model constraints. Our proposed approach is described in section IV. Simulation results are provided and analyzed in section V. Finally, section VI draws the conclusions of our paper.

\textit{Notations}: In this paper, $x$ is a scalar, $\boldsymbol{x}$ or $\boldsymbol{X}$ is a vector or matrix. $\boldsymbol{X}^T$, $\boldsymbol{X}^*$, $\boldsymbol{X}^{-1}$ denote the transpose, Hermitian, and pseudo-inverse of $\boldsymbol{X}$, respectively.  $|.|$ denotes the modulus of a complex number
 and $||.||^2$ denotes the $l_2$-norm of
the vector. $j=\sqrt{-1}$ denotes the imaginary unit.



\section{System Model}

Consider a downlink communication between a base station (BS) equipped with a uniform linear array (ULA) of $M$ antennas and $K$ users. We suppose that the direct link between the BS and users is obstructed due to the presence of obstacles. To overcome this problem, a RIS consisting of ULA with $N$ reflecting elements, deployed on a drone, is employed to redirect the signals towards obstructed users on the ground. We also suppose that the drone travels from an initial point to a target destination while simultaneously establishing a link between the BS and the obstructed users along its trajectory. An illustration of the system model is shown in Figure \ref{figSys}. 


\begin{figure}
    \centering
    \includegraphics[scale=0.4]{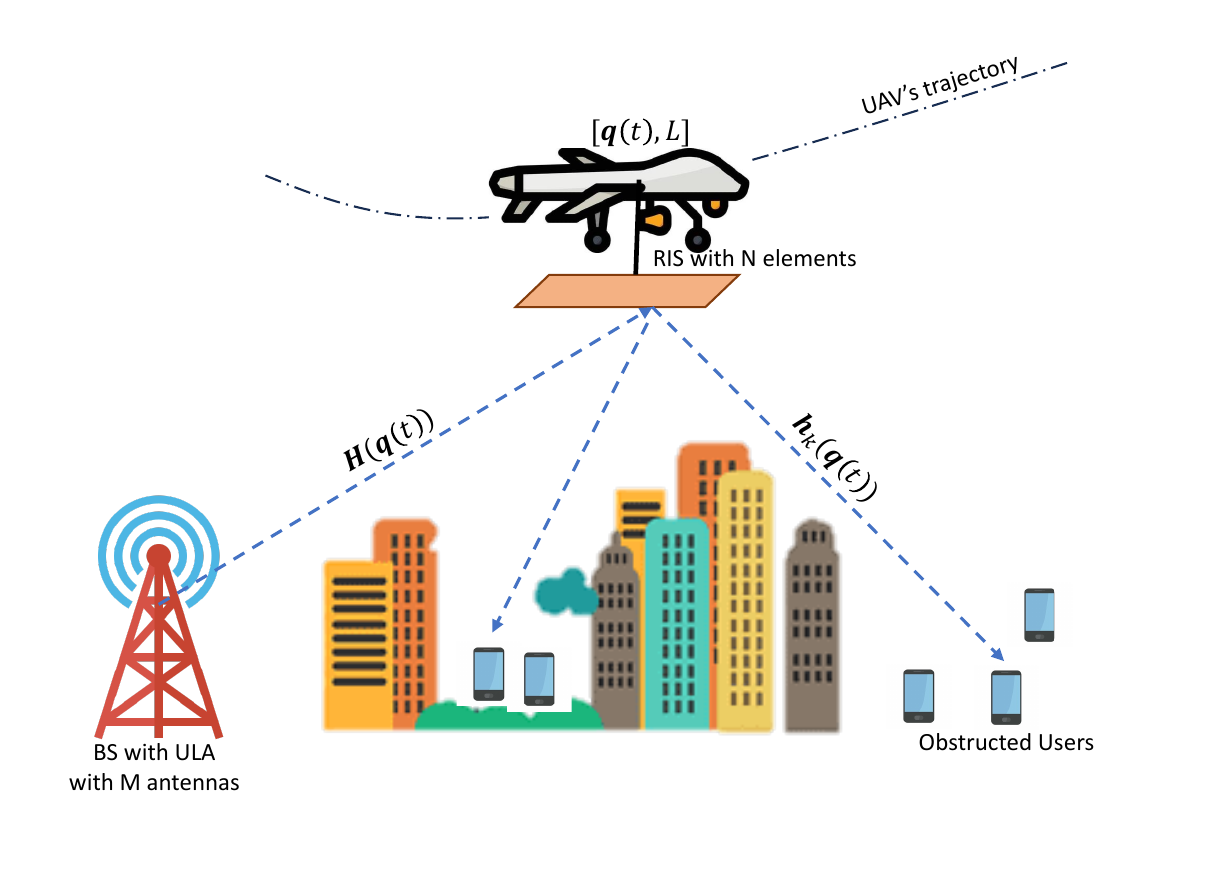}
    \caption{System model.}
    \label{figSys}
\end{figure}
Without loss of generality, we assume a 3D Cartesian coordinate system where the BS is located at the origin. The ARIS lies within a horizontal plane that is parallel to the XY plane. We assume that during its trajectory, the movement of the UAV does not change the horizontal orientation of the RIS \footnote{This is possible using omnidirectional motor aerial vehicles (MRAVs) \cite{Omnidirectional2024}.}.
We also suppose that the UAV flies horizontally at a fixed altitude $L$ and denote its position at time $t$ by $\boldsymbol{\rho}^{\rm UAV}(\boldsymbol{q}(t))=[\boldsymbol{q}(t)^T, L]^T$ with $\boldsymbol{q}(t)$ its 2D coordinates. Similarly, we denote by $\boldsymbol{\rho}^{\rm User}_{k}$ the 3D coordinates of user $k$ which we assume fixed over the time horizon $[0,T]$ for a given $T>0$.

\subsection{Channel Model}

We assume that the communication links between the BS and the UAV, and between the UAV and the ground users are predominantly characterized by the line-of-sight (LoS) channels. As a result, the time-varying channel between the BS and the UAV at time $t$, $H(\boldsymbol{q}(t)) \in \mathbb{C}^{N\times M}$, and between the UAV and user $k$, $h_k(\boldsymbol{q}(t))\in \mathbb{C}^{1\times N}$, can be modeled as follows \cite{jeon2022energy}
\begin{equation}\label{H}
     \boldsymbol{H}\!(\boldsymbol{q}(t))\!= 
   \alpha\big(\boldsymbol{q}(t)\big) \boldsymbol{{\Gamma}}_{RIS}\Big(\phi\big(\boldsymbol{q}(t)\big)\Big) \boldsymbol{\Gamma}_{BS}^*\Big(\kappa\big(\boldsymbol{q}(t)\big)\Big), 
\end{equation}
\begin{equation}\label{h}
    \boldsymbol{h}_k^*(\boldsymbol{q}(t))=\beta_k\big(\boldsymbol{q}(t)\big)\boldsymbol{\Gamma}_{RIS}^*\big(\mu_k\big(\boldsymbol{q}(t)\big)\big),
\end{equation}

\noindent with
$
     \alpha\big(\boldsymbol{q}(t)\big)=\frac{\sqrt{l_0}}{{\parallel}\boldsymbol{\rho}^{\rm UAV}(\!\boldsymbol{q}(t)\!){\parallel}_{2}} 
     e^{j\phi_{H}} e^{-j \frac{2\pi {\parallel}\boldsymbol{\rho}^{\rm UAV}(\boldsymbol{q}(t)){\parallel}_{2}}{\lambda}},$
and $   \beta_k\big(\!\boldsymbol{q}(t\!)\!\big)\!=\!\frac{\sqrt{l_0}}{{\parallel}\boldsymbol{\rho}^{\rm UAV}(\boldsymbol{q}(t))-\boldsymbol{\rho}_k^{\rm User}{\parallel}_{2}}\! 
 e^{j\phi_{h}}\! e^{-j \frac{2\pi {\parallel}\boldsymbol{\rho}^{\rm UAV}(\boldsymbol{q}(t))-\boldsymbol{\rho}_k^{\rm User}{\parallel}_{2}}{\lambda}},$
\\

\noindent where $l_{0}$ is the reference path loss at a link distance of
$1$m, $\phi_{H}$ and $\phi_{h}$ are independent random phases in $[0,2\pi]$,  $\boldsymbol{\Gamma}_{RIS}(.) \in \mathbb{C}^{N\times 1}$ and $\boldsymbol{\Gamma}_{BS}(.) \in \mathbb{C}^{M\times 1}$ are the array response of the RIS and the BS, respectively, which depend on the angle of arrival (AoA) between the BS and the RIS at time $t$ given by $\phi\big(\boldsymbol{q}(t)\big) \in [0,2\pi]$, the angle of 
departure
(AoD) between the BS and the RIS denoted by $\kappa\big(\boldsymbol{q}(t)\big) \in [0,2\pi]$, and the AoD between the RIS and user $k$ given by $\mu_k\big(\boldsymbol{q}(t)\big) \in [0,2\pi]$. Specifically, the response arrays are modeled as follows

\vspace{0.3cm}
$\boldsymbol{\Gamma}_{BS}(.)=[1, e^{-j\frac{2\pi d_{BS}sin(.)}{\lambda} },\dots,e^{-j\frac{2\pi (M-1)d_{BS}sin(.)}{\lambda} }]^T,$
 
$\boldsymbol{\Gamma}_{RIS}(.)=[1, e^{-j\frac{2\pi d_{RIS}sin(.)}{\lambda} },\dots,e^{-j\frac{2\pi (N-1)d_{RIS}sin(.)}{\lambda} }]^T,
$
with $d_{BS}$ and $d_{RIS}$ are the element separations in BS and RIS, respectively. 

For the sake of tractability, we assume that the channels $\boldsymbol{H}(\boldsymbol{q}(t))$, and $\{\boldsymbol{h}_k(\boldsymbol{q}(t))\}_{k=1}^K$ are known at the BS, which can be achieved using one of the techniques in \cite{wei2021channel}. We also suppose that the BS uses a frequency division multiple access (FDMA) and thus there is no interference at the receiver. 
\subsection{Signal-to-Noise-Ratio}

Let $\eta_k\big(\boldsymbol{q}(t),\boldsymbol{\Theta}(t)\big)$ be the signal-to-noise-ratio (SNR) of user $k$ at time $t$, which depends on the position of the ARIS $\boldsymbol{q}(t)$ and the phase shift $\boldsymbol{\Theta}(t)$. The SNR is described formally, for given $t\geq 0$ and $k \in \{1,..,K \}$, as follows:
\begin{equation}\label{InstantRatee}
   \eta_k\big(\boldsymbol{q}(t),\boldsymbol{\Theta}(t)\big)=\!\frac{P_k|\boldsymbol{h}^*_k(\boldsymbol{q}(t))\boldsymbol{\Theta(t)} \boldsymbol{H}(\boldsymbol{q}(t))\boldsymbol{w}_k(\boldsymbol{q}(t))|^2}{\sigma^2},
\end{equation}
where $\sigma^2$ is the variance of an additive white Gaussian noise, $P_k$ and $\boldsymbol{w}_k(\boldsymbol{q}(t))  \in \mathbb{C}^{M\times 1}$ are the transmit power and the precoding vector intended to user $k$, and $\boldsymbol{\Theta}(t)=diag\Big((e^{j\theta_n(t)})_{n=0}^{N-1}\Big)$ is a diagonal phase shift matrix at time $t$ with $\theta_n(t)\in[0,2\pi]$ the phase shift of RIS element $n$.

We assume that a maximum ratio transmission (MRT) strategy is adopted. Consequently, the optimal precoding vector is given by lemma 1 in \cite{jeon2022energy}, as follows
\begin{equation}\label{v}  \boldsymbol{w}_k(\boldsymbol{q}(t))=\boldsymbol{\Gamma}_{BS}\Big(\kappa\big(\boldsymbol{q}(t)\big)\Big)/{\parallel}\boldsymbol{\Gamma}_{BS}\Big(\kappa\big(\boldsymbol{q}(t)\big)\Big){\parallel}_2.
\end{equation}

Note that, under the MRT strategy, the precoding vector is independent of the users' positions and is determined by the AoD between the BS and the RIS, as demonstrated in~\cite{jeon2022energy}.

\subsection{Energy Efficiency}

In this paper, we are interested in the energy efficiency of a fixed-wing drone traveling from an initial point to a target destination. While traveling, the drone carrying a RIS establishes a link between the BS and the obstructed users along with its trajectory. Our aim is to maximize the transmitted data received by the users during the drone's flight while minimizing its consumed energy. Specifically, the total amount of bits that can be transmitted from the BS to a user $k$ during the time horizon $[0,T]$, ${R}_k(\!\boldsymbol{q},\boldsymbol{\Theta}\!)$, is dependent on the UAV's trajectory and phase shift functions, and can be expressed as
\begin{equation} \label{rate}     
    {R}_k(\!\boldsymbol{q},\boldsymbol{\Theta}\!)= \!\frac{\mathcal{B}}{K}\!\!\int_0^T\!\!\log_2\!\left(\!1+\eta_k\big(\!\boldsymbol{q}(t),\boldsymbol{\Theta}(t)\!\big)\right)dt,
\end{equation}
where $\mathcal{B}$ is the transmission bandwidth equally divided between the $K$ users.

Moreover, the total energy consumption of the UAV involves two main components. The first component is associated with phase switching and results from the switching of RIS phases. The second component, known as the propulsion energy, is essential to support the mobility of the drone. In practical scenarios, the energy consumed due to phase switching is significantly smaller compared to the UAV's propulsion energy. Therefore, for the purpose of this paper, we neglect the phase switching-related energy and focus on the propulsion energy during time period T, which is given by \cite{zeng2017energy}

\begin{equation}\label{energy}
\begin{array}{ll}
   {E}(\boldsymbol{q})&=\!\int_{0}^{T}\!\! \Bigg( \!c_{1}\!\!\parallel\! \boldsymbol{v}(t)\!\parallel^{3}_2\!+\!\frac{c_{2}}{\parallel\! \boldsymbol{v}(t)\!\parallel}_2 \!(1+\frac{\parallel \!a(t)\!\parallel^{2}_2-\frac{(\boldsymbol{a}^{T}(t)\boldsymbol{v}(t))^{2}}{\parallel \boldsymbol{v}(t) \parallel^{2}_2}}{g^{2}}\! )\! \Bigg)dt \\&+ \frac{1}{2}m\big(\parallel \boldsymbol{v}(T)\parallel^{2}_2-\parallel \boldsymbol{v}(0)\parallel^{2}_2\big),
    \end{array}
\end{equation}
where $t \mapsto\boldsymbol{v}(t)\ \in \mathbb{R}^2 $ is the velocity vector of the drone and $t\mapsto \boldsymbol{a}(t) \in \mathbb{R}^2$  
its acceleration vector. The velocity and accelerations of the drone evolve according to the following dynamical model\footnote{Let us mention that the proposed framework in this paper can be generalized to any other type of drones with more complex nonlinear dynamics~\cite{elmokadem2021towards}.}: 
\begin{align}
\label{eqn:cont}
    \boldsymbol{v}(t)=\dot {\boldsymbol{q}}(t),\;
    \ \boldsymbol{a}(t)=\ddot {\boldsymbol{q}}(t).
\end{align}
The parameters $c_{1}$ and $c_{2}$ are associated with factors such as weight, wing area, and air density. Additionally, $g$ represents the gravitational acceleration, and $m$ corresponds to the mass of the drone.

Finally, the energy efficiency of the ARIS is given by the sum of the transmitted amount of bits of all users during the time interval $[0,T]$ over its consumed energy, i.e., 

\begin{equation}
    EE(\boldsymbol{q},\boldsymbol{\Theta})=\sum \limits_{k=1}^K {R}_k(\boldsymbol{q},\boldsymbol{\Theta})/{E}(\boldsymbol{q}).
\end{equation}

In the next section, we set the mathematical formulation of the considered optimization problem under the constraints of the dynamic model of the drone.

\section{Problem Formulation}
Our main objective is to jointly determine the UAV's trajectory, its acceleration and velocity profiles, as well as the phase shifts during the UAV's mission in order to maximize its energy efficiency. The problem is formulated considering the constraints imposed by the UAV's dynamic model, constraints on the maximum acceleration and velocity, and the constraints on the initial and final positions. Consequently, our optimization problem is formulated as follows

\begin{maxi!}|s|
{\boldsymbol{q},\boldsymbol{\Theta},\boldsymbol{v},\boldsymbol{a}}{ EE(\boldsymbol{q},\boldsymbol{\Theta})}{}{}
\addConstraint{ \theta_n(\!t\!) \in [0\!,\!2\pi\!] \; \forall n \in \{\!0,\!..\!,\!N\!\!-\!1\!\}, \forall t \in[0,T] \label{Const2}}{}{}
\addConstraint{\boldsymbol{v}(t)=\dot {\boldsymbol{q}}(t) \quad \forall t \in[0,T]\label{Const3}}{}{}
\addConstraint{\boldsymbol{a}(t)=\ddot {\boldsymbol{q}}(t) \quad \forall t \in[0,T]\label{Const4}}{}{}
\addConstraint{\parallel\!\boldsymbol{v}(t)\!{\parallel}_2\!\leq \!v^{\rm max},\parallel\!\!\boldsymbol{a}(t)\!{\parallel}_2\!\!\leq\! a^{\rm max}\; \forall t \in[0,\!T]\label{Const5}}{}{}
\addConstraint{\boldsymbol{q}(0)=\boldsymbol{u}_0, \boldsymbol{q}(T)=\boldsymbol{u}_T.\label{Const7}}{}{}
\end{maxi!}
where $v^{max}$, $a^{max}$ are the maximum velocity and maximum acceleration, respectively, and $\boldsymbol{u}_0$ and $\boldsymbol{u}_T$ are initial and target positions of the UAV.

Directly solving the studied optimization is significantly challenging for many reasons. First, the optimization variables $\boldsymbol{q}$ and $\boldsymbol{\Theta}$ are tightly coupled. Second,
it necessitates the optimization of the continuous function $\boldsymbol{q}(t)$, along with its first- and second-order derivatives $\boldsymbol{v}(t)$ and $\boldsymbol{a}(t)$. Finally, the objective function in $EE\big(\boldsymbol{q},\boldsymbol{\Theta}\big)$ is defined as the ratio of two integrals, both of which lack closed-form expressions. In the following, we present an effective solution that addresses the problem through the use of convex approximation and EMPC.
\section{Phase Shift Optimization and Trajectory Design for an ARIS}
In this section, we start by transforming the continuous-time problem above into a discrete-time problem. Then, we introduce an efficient solution that involves solving the phase shift optimization as a convex approximation. This yields a closed-form expression for the phase shift, which is subsequently integrated into the trajectory design problem. The resulting optimization is then solved using an EMPC-inspired approach.

\subsection{Discrete-time optimization problem}

 Consider the continuous-time linear dynamical system in (\ref{eqn:cont}) and consider a finite time step $\delta_{t}>0$, we have the following results based on the first- and second-order Taylor approximations. By considering $t=s\delta_{t}$ and $s=0,1,...,S+1$, we can write 
\begin{equation}\label{Veloc}
    \boldsymbol{v}[s+1]= \boldsymbol{v}[s]+\boldsymbol{a}[s]\delta_{t} \hspace*{2,53cm} \forall s \in \{0,..,S\},
\end{equation}
\begin{equation}\label{Acc}
    \boldsymbol{q}[s+1]=\boldsymbol{q}[s]+\boldsymbol{v}[s]\delta_{t} +\frac{1}{2}\boldsymbol{a}[s]\delta_{t}^{2}, \quad \forall s \in \{0,..,S\}.
\end{equation}

We also discretize the integrals in equations (\ref{rate}) and (\ref{energy}). Consequently, the discrete transmitted amount of data and energy are given by  

\begin{equation} \label{rate2}     
    \bar{R}_k(\!\boldsymbol{q},\boldsymbol{\Theta})\!= \!\frac{\mathcal{B}}{K}\!\!\sum_{s=0}^{S}\!\!\log_2\!\left(\!1+\eta_k\big(\!\boldsymbol{q}[s],\boldsymbol{\Theta}[s]\!\big)\right),
\end{equation}
\begin{equation}\label{energy2}
\begin{array}{ll}
   \bar{E}(\boldsymbol{q})=&\!\!\sum\limits_{s=0}^{S}\!\! \Bigg( \!\!c_{1}\!\parallel \!\boldsymbol{v}[s]\!\parallel^{3}_2+\frac{c_{2}}{\parallel \boldsymbol{v}[s]\parallel}_2\! (1+\frac{\parallel a[s]\parallel^{2}_2-\frac{(\boldsymbol{a}^{T}[s]\boldsymbol{v}[s])^{2}}{\parallel \boldsymbol{v}[s] \parallel^{2}_2}}{g^{2}} )\!\! \Bigg) +\\& \frac{1}{2}m\big(\parallel \boldsymbol{v}[S]\parallel^{2}_2-\parallel \boldsymbol{v}(0)\parallel^{2}_2\big).
    \end{array}\end{equation}
Accordingly, the discrete energy efficiency is written as 
\begin{equation}
       \bar{EE}(\boldsymbol{q},\boldsymbol{\Theta})=\sum \limits_{k=1}^K \bar{R}_k(\boldsymbol{q},\boldsymbol{\Theta})/\bar{E}(\boldsymbol{q}).
\end{equation}
We also discretize constraint (\ref{Const5}) which becomes
\begin{equation}\label{equaV}
    \parallel v[s]\parallel_2 \leq v^{\rm max}, \text{and} \parallel a[s]\parallel_2 \leq a^{\rm max} \;\forall s \in \{0,.., S\}
\end{equation}

Finally, under the discrete-time approximation, the studied optimization problem can be rewritten as follows.
\begin{maxi!} |s|
   {\boldsymbol{q} , \!\boldsymbol{\Theta}, \!\boldsymbol{v},\!\boldsymbol{a}}
   {\bar{EE}(\boldsymbol{q},\boldsymbol{\Theta})\label{maxii4}}{\label{maxiii4}}{}
   \addConstraint{ (\!\ref{Veloc}\!),\!(\!\ref{Acc}\!),\! (\!\ref{Const2}\!),(\!\ref{Const5}\!),\! (\!\ref{equaV}\!),\!(\!\ref{Const7}\!).  \label{Const16}}{}{}
\end{maxi!}

Directly solving the above optimization problem is
still challenging due to the intricate coupling between the trajectory and phase shift variables. To tackle this, we decouple the phase shift and trajectory design problems.

\subsection{Phase Shift Optimization}
In this subsection, we assume that the position of the UAV as well as its velocity and acceleration are known. This implies, that solving problem (\ref{maxiii4}) reduces to maximizing the discrete amount of transmitted data in (\ref{rate2}) with respect to the phase shift matrix at a given discrete time $s$. For ease of notation, we drop the dependency on $\boldsymbol{q}[s]$ in this subsection. Accordingly, the problem becomes 
\begin{equation} \label{maxii3}
   \max_{\boldsymbol{\Theta}[s]\in [0,2\pi]}
   {\sum \limits_{k=1}^K\bar{R}_k(\boldsymbol{\Theta}[s] )}.
\end{equation}

To address this optimization problem, we first simplify the expression of the SNR under the MRT strategy. The resulting SNR expression, as a function of the phase shift matrix, is presented in the following lemma.  
\begin{lemma}\label{etalemma}
    Under MRT strategy, the SNR experienced by user $k$ at time $s$ is given by 
 \begin{equation}
\eta_k\big(\boldsymbol{\Theta}[s]\big)\!=\!C_k\!\!\left(\!N\!+\!\!\sum \limits_{\substack{(n,m) \\ n\neq m}}\!\!\!cos\Big(B_n^k\big(\boldsymbol{\Theta}[s]\big)\!-\!B_m^k\big(\boldsymbol{\Theta}[s]\big)\Big)\!\!\right) 
\end{equation}
where   $C_k=\frac{P_k}{\sigma^2}\left|\alpha\beta_k\right|^2$and $B_n^k\big(\!\boldsymbol{\Theta}[s]\!\big)\!=\!-\frac{2\pi}{\lambda}nd_{RIS}\Big(\!sin\big(\mu_k\big)\!-\!sin\big(\Phi\big)\!\Big)\!+\!\theta_n[s]$.
\end{lemma}

\begin{proof}
 By using the expression of $\boldsymbol{w}_k$ in equation (\ref{v}), and plugging equations (\ref{H}) and (\ref{h}) in equation (\ref{InstantRatee}), we obtain 



\begin{equation}\label{eta2}
    \eta_k\big(\!\boldsymbol{\Theta}[s]\!\big)\!\!=\!\!\frac{P_k}{\sigma^2}\!\!\left|\alpha\beta_k\!\right|^2\!\!
      \left|\!\sum \limits_{n=0}^{N-1}\!\!\!e^{-j\!\frac{2\pi}{\lambda}nd_{RIS}\!\big(\!sin(\!\mu_k\!)\!-\!sin(\!\Phi\!)\!\big)\!+\!j\theta_n[s]}\!\right|^2.    
\end{equation}
We denote $C_k=\frac{P_k}{\sigma^2}\left|\alpha\beta_k\right|^2$ and $B_n^k\big(\!\boldsymbol{\Theta}[s]\!\big)\!=-\!\frac{2\pi}{\lambda}nd_{RIS}\big(\!sin(\!\mu_k\!)\!-\!sin(\!\Phi\!)\!\big)\!+\!\theta_n[s]$, we have 
\begin{equation}\label{eta3}
\begin{array}{ll}
    \eta_k\big(\!\boldsymbol{\Theta}[s]\!\big)\!\!\!\!\!\!&=\!C_k
      \left|\sum \limits_{n=0}^{N-1}\!\!e^{jB_n^k\big(\!\boldsymbol{\Theta}[s]\!\big)}\right|^2\\
      &\stackrel{(a)}{=}\!C_k
      \!\!\!\left(\!\!\left(\!\sum \limits_{n=0}^{N-1}\!\!cos\Big(\!B_n^k\big(\!\boldsymbol{\Theta}[s]\!\big)\!\Big)\!\!\right)^2\!\!\!\!+\!\!\left(\!\sum \limits_{n=0}^{N-1}\!\!sin\Big(\!\!B_n^k\big(\!\boldsymbol{\Theta}[s]\!\big)\!\Big)\!\!\right)^2\right)\\
      &\stackrel{(b)}{=}C_k\!\!\left(\!\!N\!\!+\!\!\!\!\sum \limits_{\substack{(n,m) \\ n\neq m}}\!\!\!cos\Big(B_n^k\big(\boldsymbol{\Theta}[s]\big)\!-\!B_m^k\big(\boldsymbol{\Theta}[s]\big)\!\Big)\!\!\right), 
\end{array}\end{equation}
where $(a)$ comes from the definition of the norm of a complex number, whereas $(b)$ stems from a direct application of trigonometric formula. 
\end{proof}
Upon substituting the expression for $ \eta_k\big(\boldsymbol{\Theta}[s]\big)$ as provided in lemma \ref{etalemma} into equation (\ref{rate2}), it is clear that the resulting function is non-concave. To tackle this problem, we introduce a tight lower bound as stated by the following lemma.

 \begin{lemma}\label{DataRateLemma}

 The transmitted data at
user $k$ during time step $s$ is lower bounded as follows  \begin{equation}\label{ineqdata}
\begin{array}{ll}
     \bar{R}_k(\boldsymbol{\Theta}[s])&\geq
   \gamma_k\!\left(\!N^2\!-\!\frac{1}{2}\sum \limits_{\substack{(n,m) \\ n\neq m}}\!\Big(B_n^k\big(\boldsymbol{\Theta}[s]\big)\!-\!B_m^k\big(\boldsymbol{\Theta}[s]\big)\Big)^2 \right)\\
   &\triangleq \bar{R}_{k}^{lb}\big(\boldsymbol{\Theta}[s]\big),
   \end{array}
\end{equation}
where $\gamma_k=\frac{\mathcal{B}C_k}{ln(2)K(1+C_kN^2)}$.
 \end{lemma}
\begin{proof}
    The inequality can be proven by using the first-order Taylor inequalities of the cosine and log functions. Specifically, 
    $\frac{x}{ln(2)(1+x)}\leq log_2(1+x), \; x\geq 0$ and $1-\frac{x^2}{2}\leq cos(x), \; x \in \mathbb{R}$.
    Accordingly, we have 
    \begin{equation}\label{ineqeta1}
  \frac{\eta_k(\boldsymbol{\Theta}[s])}{ln(2)(1+C_k  N^2)}\leq log_2(1+\eta_k(\boldsymbol{\Theta}[s])) 
\end{equation}
and
\begin{equation}\label{ineqeta2}
   1-\frac{1}{2}\!\left(B_n^k\big(\boldsymbol{\Theta}[s]\big)\!-\!B_m^k\big(\boldsymbol{\Theta}[s]\big)\right)^2\! \!\leq \!cos\Big(\!B_n^k\big(\boldsymbol{\Theta}[s]\big)\!-\!B_m^k\big(\boldsymbol{\Theta}[s]\big)\!\Big). 
\end{equation}

By combining the two inequalities with the expression of the transmitted amount of data, we obtain inequality (\ref{ineqdata}).
\end{proof}

We note that the lower-bound is tight for $B_n^k\big(\boldsymbol{\Theta}[s]\big)\equiv 0 \;mod(2\pi)$ for all $n$ and $k$. Hence, to find the phase shift matrix, we solve the following approximated optimization

\begin{equation} \label{maxii44}
   \max_{\boldsymbol{\Theta}[s] \in [0,2\pi]}
   {\sum\limits_{k=1}^K\sum \limits_{\substack{(n,m) \\ n\neq m}}\!\gamma_k\Big(B_n^k\big(\boldsymbol{\Theta}[s]\big)\!-\!B_m^k\big(\boldsymbol{\Theta}[s]\big)\Big)^2}.
\end{equation}

The following theorem provides the optimal phase shift matrix for the convex approximation.

\begin{theorem}
    The optimal phase shift matrix that minimizes problem (\ref{maxii44}) is given by 

    \begin{equation}\label{eq:optimalTheta}
        \boldsymbol{\Theta}^*[s]=\left(\sum \limits_{k=1}^K\gamma_k A^TA\right)^{-1}\left(\sum \limits_{k=1}^K\gamma_k A^T\boldsymbol{b}^k\right),
    \end{equation}
    where $A\in \mathbb{R}^{\frac{N(N-1)}{2}\times N}=(A^1,\dots,A^N)^T$, and $$A^i=\left[\boldsymbol{0}_{(N-i)\times (i-1)}, \bold{1}_{(N-i) \times 1}, -\boldsymbol{I}_{(N-i)\times (N-i)}\right],$$ with $\boldsymbol{0}_{(N-i)\times (i-1)}$ the matrix of size $(N-i)\times (i-1)$ with zeros everywhere, $\bold{1}_{(N-i) \times 1}$ is the column vector of size $N-i$ with entries equal to $1$, and $\boldsymbol{I}_{(N-i)\times (N-i)}$ is the identity of size $N-i$, and $\boldsymbol{b}^k$ is matrix of size $N\times N$ where $$b^k_{n,m}=\!\frac{2\pi}{\lambda}(m-n)d_{RIS}\big(\!sin(\!\mu_k\!)\!\big).$$
\end{theorem}
\begin{proof}
 Using the new notations introduced in the theorem, the objective function in (\ref{maxii44}) can be written as
     \begin{equation*}
        \sum \limits_{k=1}^K \gamma_k \left\|A \boldsymbol{\Theta}[s] - \boldsymbol{b}^k \right\|^2.
    \end{equation*}
    By derivating the previous quantity and solving the gradient equal to zero, we get equation (\ref{eq:optimalTheta}).
\end{proof}
Note that the optimal phase in equation (\ref{eq:optimalTheta}) is a function of the drone's position. In the following, we reestablish the dependency on $\boldsymbol{q}[s]$. Specifically, we denote
\begin{equation}\label{EquaThetaf}
    \boldsymbol{\Theta}^*[s]=f(\boldsymbol{q}[s]),
\end{equation}
with $f(\boldsymbol{q}[s])=\!\!\left(\!\sum \limits_{k=1}^K\!\gamma_k(\boldsymbol{q}[s]) A^T\!A\!\right)^{-1}\!\!\!\!\left(\sum \limits_{k=1}^K\!\gamma_k(\boldsymbol{q}[s]) A^T\boldsymbol{b}^k(\boldsymbol{q}[s])\!\!\right)$.

\subsection{Design of the Control Strategy}

To address the studied optimization problem, we first plug the expression of the phase shift matrix as a function of $\boldsymbol{q} [s]$, described by equation (\ref{EquaThetaf}), in problem (\ref{maxiii4}). Accordingly, the problem reduces to a trajectory optimization described as follows \begin{maxi!} |s|
   {\boldsymbol{q} ,\!\boldsymbol{v},\!\boldsymbol{a}}
   {\bar{EE}(\boldsymbol{q})\label{maxii}}{\label{maxiii}}{}
   \addConstraint{ (\!\ref{Veloc}\!),\!(\!\ref{Acc}\!),(\!\ref{Const5}\!),\! (\!\ref{equaV}\!),\!(\!\ref{Const7}\!).  \label{Const16}}{}{}
\end{maxi!}

    To tackle this problem, we use a solution strategy inspired by the classical EMPC approach in control theory \cite{ellis2014tutorial}. The proposed control strategy operates over a finite prediction horizon, which is a sequence of future time steps. The goal is to find control inputs that minimize the cost function over each sequence of time steps. The optimization is solved at each time step, and the control input for the current time is implemented. The system is then advanced by one-time step, and the optimization is solved again. This process is repeated in real-time, allowing the control strategy to adapt to changing conditions. In order to provide the new formulation of (\ref{maxiii}), we introduce some notations.

First, the linear system in (\ref{Veloc})-(\ref{Acc}) can be written as follows
\begin{equation}
\label{eqn:model}
\boldsymbol{x}[s+1]=\boldsymbol{W}\boldsymbol{x}[s]+\boldsymbol{Z}\boldsymbol{a}[s],
\end{equation}
with $\boldsymbol{W}=\begin{bmatrix}
1 & 0 \\
\delta_t & 1 
\end{bmatrix}$, $\boldsymbol{Z}=\begin{bmatrix}
\delta_t \\
\frac{\delta_t^2}{2} 
\end{bmatrix}$, and the state variable  $\boldsymbol{x}[s]=[\boldsymbol{v}[s],\boldsymbol{q}[s]] \in \mathbb{R}^{4\times 1}$ consists of the velocity and the position in the 2D space at the discrete time instant $s \in \{0,..,S\}$, and the input $\boldsymbol{a}[s] \in \mathbb{R}^{2\times 1}$ corresponding to the acceleration of the UAV.
\begin{figure}
    \centering
   \includegraphics[scale=0.3]{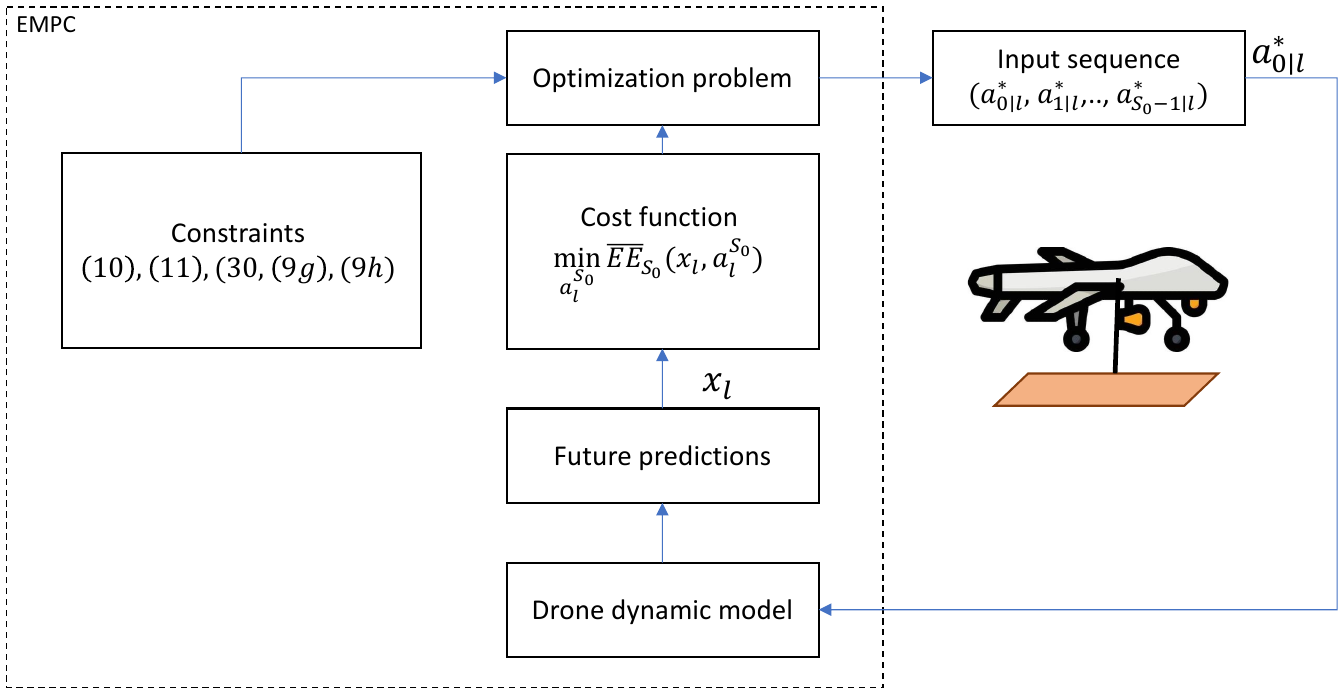}
    \caption{Illustration of the proposed control approach.}
    \label{fig:EMPC}
\end{figure}
In the following, we use the notation $\mathbf{a}=\left(\boldsymbol{a}[0], \boldsymbol{a}[1], \ldots\right)$ for a finite or infinite control sequence, and the notation $\Phi\left(s, \boldsymbol{x}_0, \mathbf{a}\right)$ for the state reached at discrete time $s$ from the initial condition $\boldsymbol{x}_0$ and under the control sequence $\mathbf{a}$. For the case of the linear system in (\ref{eqn:model}), we can obtain an explicit expression for $\Phi\left(s, \boldsymbol{x}_0, \mathbf{a}\right)$ using a recursion as follows, $\forall s \in \{1,..,S\}$
\begin{equation}\label{phi}
    \Phi\left(s, \boldsymbol{x}_0, \mathbf{a}\right)=\boldsymbol{W}^s\boldsymbol{x}[0]+\sum_{k=0}^{s-1}\boldsymbol{W}^k\boldsymbol{Z}\boldsymbol{a}[s-1-k]. 
\end{equation}
Notice that $\Phi\left(s, \boldsymbol{x}_0, \mathbf{a}\right)=[\boldsymbol{v}[s],\boldsymbol{q}[s]]$, the position and velocities at time $s$ can be written explicitly as function of the initial conditions $\boldsymbol{q}[0]$ and $\boldsymbol{v}[0]$ and the sequence of control input $\boldsymbol{a}[0],\boldsymbol{a}[1], \ldots, \boldsymbol{a}[s-1]$. Therefore, the number of decision variables of the optimization problem in (\ref{maxiii}) can be reduced to finding the sequence $\boldsymbol{a}[0],\boldsymbol{a}[1], \ldots, \boldsymbol{a}[s-1]$, with constraints (\ref{Veloc}),(\ref{Acc}) replaced by equation (\ref{phi}).



\begin{figure*}
\begin{tabular}{lll}
  \hspace{-0.5cm}  \includegraphics[scale=0.3]{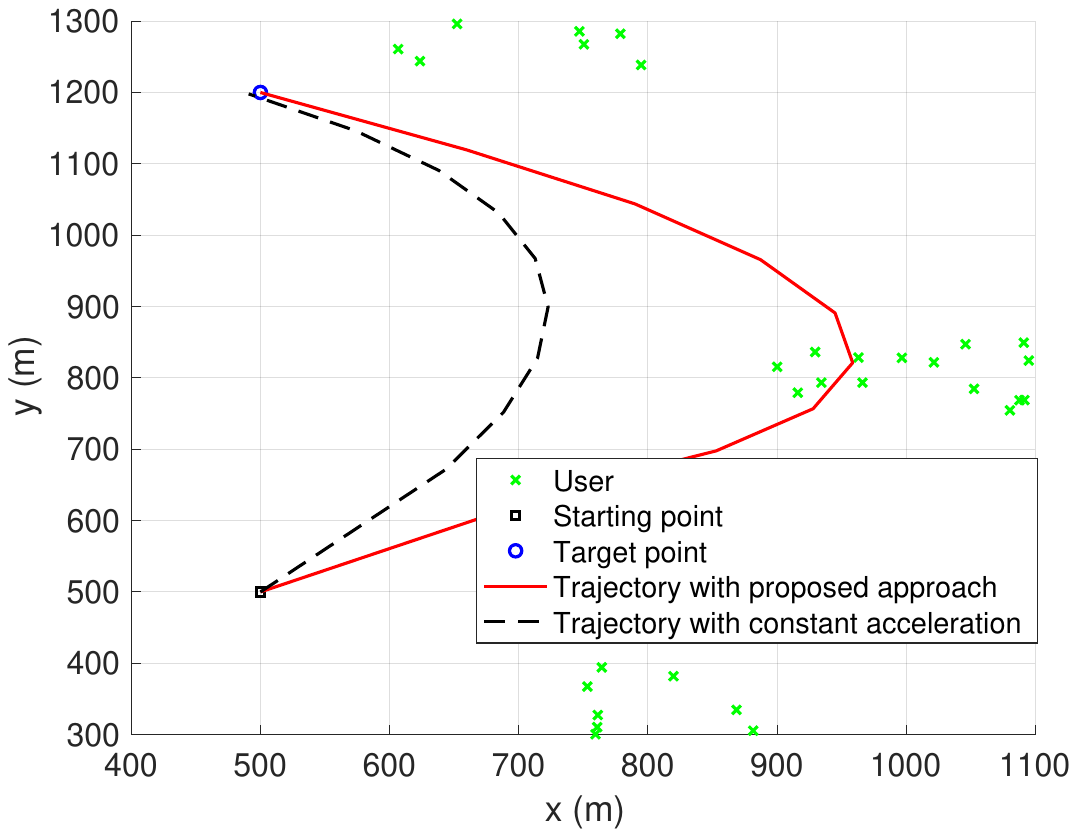}\hspace{-1cm} & \includegraphics[scale=0.3]{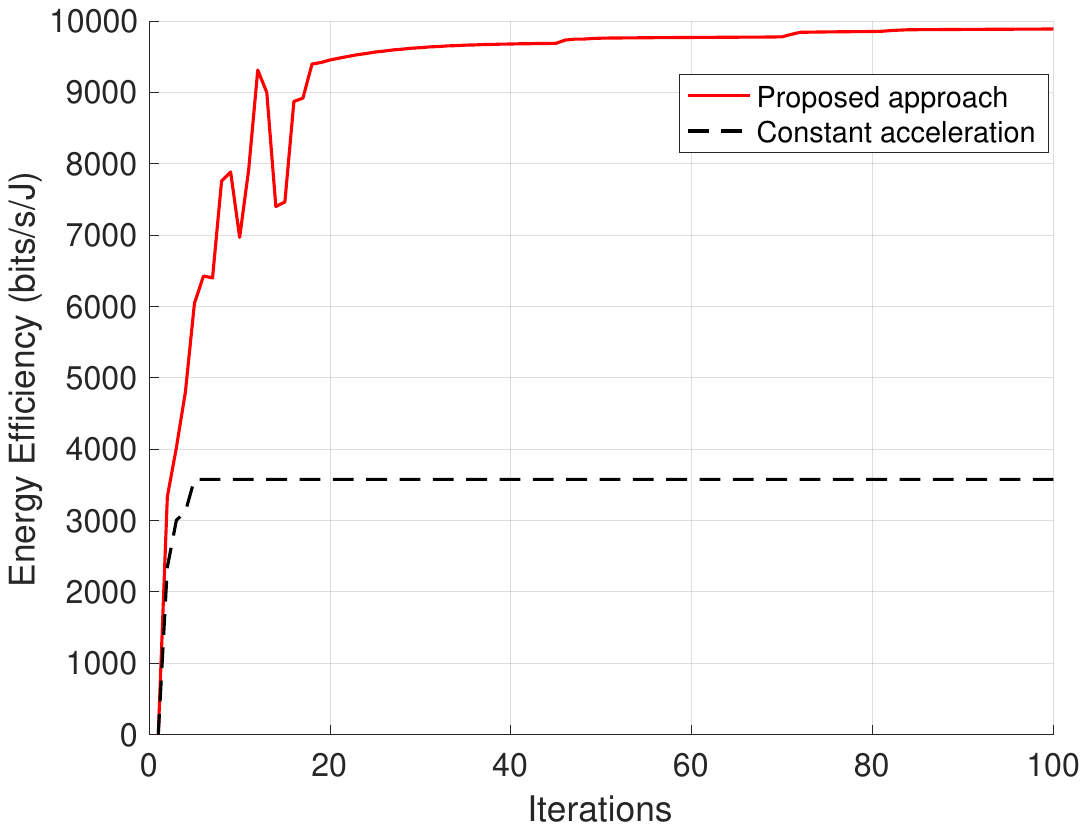}\hspace{-1.2cm} & \includegraphics[scale=0.3]{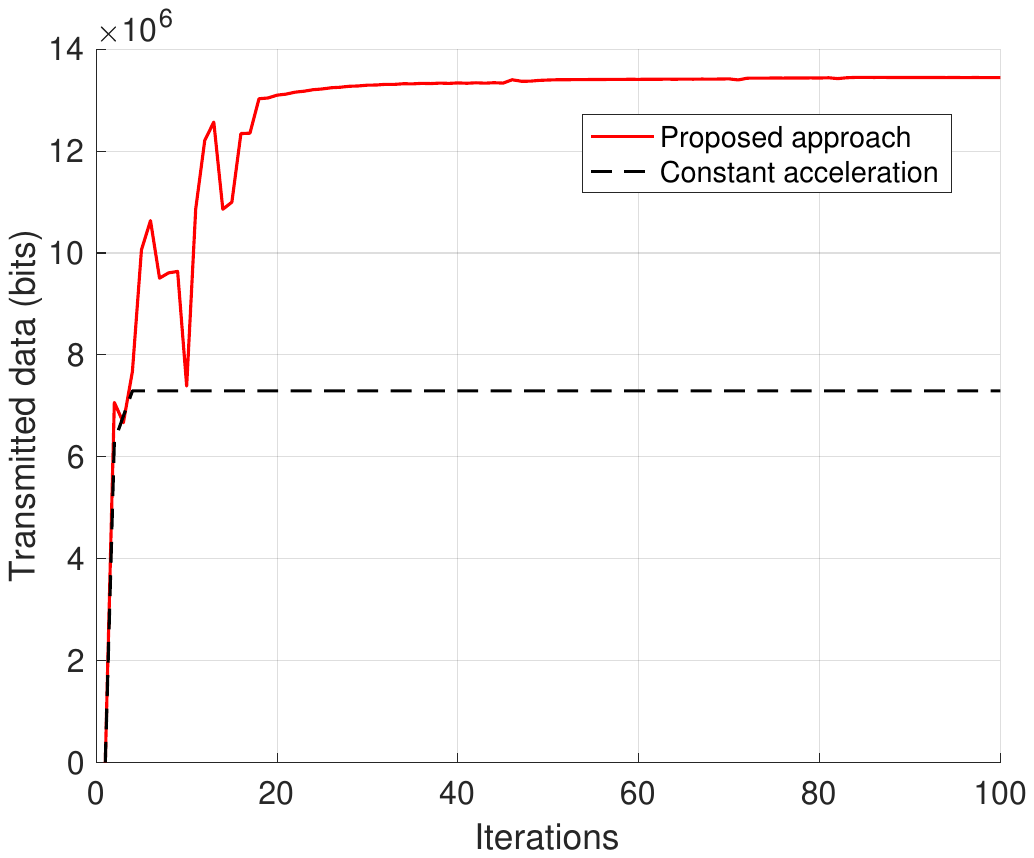} 
 \vspace{-2cm}   
\end{tabular}
\caption{(a) ARIS trajectory, (b) Energy efficiency vs iterations, (c) Transmitted amount of data.}
\label{fig:Energy}
\end{figure*}

Given $1\leq S_0\leq S$ as prediction horizon, a control sequence computed at sampling time $s$ will be denoted as $\mathbf{a}_s^{S_0}=\left(\boldsymbol{a}_{0 \mid s}, \ldots, \boldsymbol{a}_{S_0-1 \mid s}\right) \in \mathbb{R}^{2 \times S_0}$. Given state $\boldsymbol{x}_s$ at discrete time instant $s$, the associated predicted trajectory at $s \in \{0,..,S_0\}$ is denoted $\mathbf{x}_s^{S_0}=\left(\boldsymbol{x}_{0 \mid s}, \boldsymbol{x}_{1 \mid s}, \ldots, \boldsymbol{x}_{S_0 \mid s}\right)$ where $\boldsymbol{x}_{k \mid s}=\Phi\left(k, \boldsymbol{x}_s, \mathbf{a}_s^{S_0}\right)$, for $k=0, \ldots, S_0$.

Considering the control system (\ref{eqn:model}), the main idea of our control scheme relies on the resolution of an optimal control problem, at each discrete instant $s \in \{0, \ldots, S_0 \}$. Starting from the optimization problem with the acceleration as an optimization variable, this finite-horizon optimization, over the finite prediction horizon $S_0$, takes the following form
\begin{align}
\label{eqn:MPC}
 \arg \min& _{\mathbf{a}_l^{S_0}} \bar{EE}_{S_0}\left(x_l, \mathbf{a}_l^{S_0}\right)\\
 &\text{s.t } \quad x_{k \mid l}=\phi\left(k, x_l, \mathbf{a}_l^{S_0}\right), k=0, \ldots, {S_0} \nonumber\\
 &\text{and } \quad (\ref{Veloc}),(\ref{Acc}),(\!\ref{equaV}\!),(\ref{Const7}), \nonumber 
\end{align}

with $\bar{EE}_{S_0}\!\!\left(\!x_l, \mathbf{a}_l^{S_0}\!\right)\!=\!\frac{\sum_{s=0}^{N-1} 
 \bar{R}_s}{\frac{1}{2}m\big(\!\!\parallel \!{v}(S\mid j)\!\parallel^{2}_2\!-\!\parallel\! {v}(0 \mid j)\!\parallel^{2}_2\!\!\big)\!+\!\!\sum_{s=0}^{{S_0}-1}\!\! 
 \!\bar{E}_s}$

where $\bar{R}_s=\sum_{k=1}^{K}\!\frac{\mathcal{B}}{K}\log_2\!\left(\!1+\eta_k\big(\!(q(s \mid j),f(q(s\mid j)\!\big)\right)$

where the last equality follows from (\ref{EquaThetaf}) and the fact that
\begin{equation*}
\begin{array}{ll}
   \bar{E}_s=\! c_{1}\!\parallel\! {v}(s \!\mid \!l)\!\parallel^{3}_2+\frac{c_{2}}{\parallel {v}(s \mid l)\parallel}_2 \!(1+\frac{\parallel {u}(s \mid l)\parallel^{2}_2-\frac{({a}^{T}(s \mid l){v}(s \mid l))^{2}}{\parallel {v}(s \mid l) \parallel^{2}_2}}{g^{2}} )  
    \end{array}
\end{equation*}

The optimization problem (\ref{eqn:MPC}) admits at least one optimizer, which is denoted by the optimal control sequence $
\mathbf{a}_l^{S_0 *}\left(x_l\right)=\left(a_{0 \mid l}^*, \ldots, a_{S_0-1 \mid l}^*\right)
$. Then, the control action fed to the system can be described as a feedback law
$\mu_{S_0}\left(x_l\right)=a_{0 \mid l}^*$.

The resolution of the optimization problem (\ref{eqn:MPC}) is repeated online, at each sampling instant leading to the following closed-loop behavior $x_{l+1}=\boldsymbol{W}x_l+\boldsymbol{Z}\mu_{S_0}\left(x_l\right)$. The optimization problem is solved numerically at each time step using a non-linear programming solver \cite{ellis2014tutorial}. This process is illustrated in Figure \ref{fig:EMPC}.

\section{Simulation Results}

To assess the performance of the proposed approach, we consider an area of $1300m\times1100m$, where $K=30$ users are distributed over three clusters as can be seen in Figure \ref{fig:Energy}(a). We assume that the UAVs fly from a starting point to a target destination at an altitude of $150m$. We assume that the maximum velocity and acceleration $V^{\rm max}=50m/s$ and $a^{\rm max}=10m/s^2$, respectively, and the initial velocity is $V[0]=30m/s$. The initial UAV position is $[500,500]$, and the target position is $[500,1200]$. The time step $\delta_t=3s$, and $S=10$. We assume that the BS transmits at a power $P_k=1Watt$ for every user $k$. We assume a bandwidth of $\mathcal{B}=10Mhz$ and the variance of the noise $\sigma^2=100dbm$.

In Figure \ref{fig:Energy}(a), we plot the trajectory of ARIS in the 2D plan. As can be seen from the figure, the trajectory obtained by the proposed approach starts from the initial point and deviates toward the cluster with the largest number of users. This strategic deviation enhances the transmitted signal for users in that cluster, contributing positively to energy efficiency. It is crucial to highlight that had the drone deviated towards the other clusters, it would have increased its energy consumption without significantly improving the transmitted data volume.  Consequently, this deviation would have led to a decrease in the drone's energy efficiency. Additionally, we present the trajectory of ARIS when the dynamic model of the drone is neglected. In this scenario, the drone is assumed to move with a constant acceleration (i.e., linear velocity).
As illustrated in the figure, the drone exhibits a slight deviation towards the region with the users' concentration, yet promptly readjusts its course to return to the intended destination. This maneuver is executed to conserve energy. Restricting the drone from adjusting its acceleration during flight leads to increased energy consumption, consequently preventing the drone from getting closer to the users. It is noteworthy that if only the energy consumption were taken into account in the optimization, the drone would have traversed a straight path connecting the initial point to the target destination.

In Figure \ref{fig:Energy}(b), we plot the energy efficiency against iterations. As illustrated by the figure, the energy efficiency increases over iterations until it reaches a maximal value. It can also be seen from the figure, that the energy efficiency obtained by the proposed approach, incorporating the dynamic model of ARIS, yields a considerably higher efficiency compared to the scenario where the dynamic model is omitted. It is important to note that the case where the dynamic model is ignored converges faster due to the reduced number of optimization variables involved in the optimization process. A similar behavior can be seen in Figure \ref{fig:Energy}(c) where we plot the amount of transmitted bits against iterations. As expected, the proposed approach ensures a higher amount of transmitted data. This can be explained by the fact that the ARIS is closer to the users which ensures a better throughput at the edge of the cell.

\section{Conclusion}
In this paper, we investigated the scenario where a drone flies from a starting point to a target destination while reflecting the signal to obstructed users through a RIS. Our main goal was to optimize the trajectory of the ARIS along with its acceleration and velocity, as well as the RIS phase shift to maximize the energy efficiency. The studied optimization is challenging and was solved using a linear state-space approximation. The phase shift matrix optimization and the trajectory design are decoupled and solved using convex approximation and EMPC, respectively. In future work, we are planning to extend the proposed work to a more complex scenario where the drone is conditioned by a probabilistic LoS, and investigate the scenario of interference at the users.  
\section*{Acknowledgment}
This document has been produced with the financial assistance of the European Union (Grant no. DCI-PANAF/2020/420-028), through the African Research Initiative for Scientific Excellence (ARISE), pilot programme. ARISE is implemented by the African Academy of Sciences with support from the European Commission and the African Union Commission. The contents of this document are the sole responsibility of the author(s) and can under no circumstances be regarded as reflecting the position of the European Union, the African Academy of Sciences, and the African Union Commission.
\balance
\bibliographystyle{IEEEbib}
\bibliography{references}

\end{document}